\documentclass[graybox]{svmult}

\usepackage{mathptmx}       
\usepackage{helvet}         
\usepackage{courier}        
\usepackage{type1cm}        
\usepackage{makeidx}         
\usepackage{graphicx}        
\usepackage{multicol}        
\usepackage[bottom]{footmisc}

\usepackage{amssymb}

\makeindex             

\begin{document}

\title*{Integral of an alpha unpredictable function.}
\author{M. U. Akhmet, N. Tilessov}

\institute{M. U. Akhmet \at Department of Mathematics, Middle East Technical University, 06800, Ankara, Turkey, \email{marat@metu.edu.tr,}}
\maketitle
\abstract{The present paper is devoted to the study of integral properties of alpha unpredictable functions. The problem of invariance of recurrence under integration is one of the most challenging and interesting in the theory of functions. Let us start with periodicity. Next, the problem was solved for quasiperiodic, almost periodic, and Poisson stable functions. All of the problems were subjected to the condition of bounded integrals. The alpha-unpredictable functions are the endpoint in the row of recurrent functions. The class is the cross-border point from regularity to chaos in the dynamics presented through the elements of the row. In the present research, we finalized the proof that any theoretical functional recurrence is invariant with respect to integration, provided the result is bounded.}

\section{Introduction}

The theory of recurrent functions holds a significant position in modern analysis, particularly in the fields of differential equations, and began in celestial mechanics \cite{PoinRec,Poin}. It provides a framework for describing a broad class of nonlinear and nonautonomous processes that exhibit regular, yet not strictly periodic, dynamics \cite{Birkhoff}. 

Beginning with the classical works of P. Bohl, E. Esclangon \cite{Bohl, Esclangon}, and culminating with H. Bohr \cite{Bohr}, the foundational concepts of quasi-periodic functions and almost periodicity were established \cite{LZ}. These concepts have found numerous applications in mathematical physics, control theory, biology, and neural networks. The first instance of non-closed recurrent dynamics, known as Poisson stability, was introduced by H. Poincaré in 1892. This area of study has become increasingly sophisticated in relation to recurrent functions. Subsequent research built upon Poincaré's ideas, filling the gap between periodicity and Poisson stability. Our current research pushes this understanding further.

In exploring chaotic problems, we discovered that previous studies portrayed chaos primarily as a "collective" phenomenon. By extending H. Poincaré's definition to include the property of alpha unpredictability in "individual" dynamics, we proved that chaos occurs in the closure of an alpha unpredictable trajectory \cite{Akhmet}. We have termed this new class of processes "ultra Poincaré chaos." The functions examined in this study represent initial points in the dynamics of shifts \cite{sell} within functional spaces. Consequently, we can assert that our research is focused on investigating chaos, particularly concerning conditions under which chaos is preserved in the presence of an integral operator. 

On methodological aspects of the research, it is worth noting that approvement of the convergence sequence in  the Poisson stability have traditionally been approached through methods of functional analysis \cite{LZ} or the method of comparability \cite{Shch}. In contrast, our study proposes a new, simple, and effective \textit{method of included intervals} \cite{Akhmet}. It suggests that to prove convergence on a bounded interval, one must identify a larger one such that convergence on this including set implies stability in the included interval. In this article, we have developed this method and expanded the number of structures that can be analyzed within this framework. To demonstrate its efficacy and potential, we have verified a known auxiliary result in the main part of our discussion. It effectively works for various types of differential equations and real-world models \cite{Akhmet,kagan1}, but the present realisation of the method is a specific one, and can be useful for further analysis of the theory of functions. 

Notably, alpha unpredictability, which encompasses Poisson stability, serves as a constructive condition for verification and is weaker than other conditions typically used for synchronization. We have written several papers demonstrating that conservative criteria are ineffective, while synchronization of divergent and convergent sequences of alpha unpredictability is indeed useful \cite{kagan1}. This finding may hold significant interest for research into chaos.

\section{ The main result}

In the present research, we consider continuous real-valued functions that are determined and uniformly bounded on the real axis.

\begin{definition} \cite{sell,LZ,Stepanov} \label{poisson}
	An uniformly continuous function \( f(t) \) is said to be \textit{Poisson stable} if there exists a strictly increasing sequence \( t_n \) that approaches infinity, such that the sequence of shifts \( f(t+t_n) \) converges uniformly to the function \( f(t) \) on every bounded interval.
\end{definition}
The set  $t_n$ is called a \textit{ convergence} or \textit{Poisson sequence}. 

\begin{definition} \cite{Akhmet} \label{alpha}A Poisson-stable function \( f(t) \) is defined as alpha unpredictable if, in addition to having a convergent sequence, there exist a strictly increasing sequence \( s_n \) that tends to infinity, and  positive numbers \( \epsilon_0 \) and \( \delta \) such that  for each index \( n \), the following condition is satisfied: 
	\[ 	|f(t + t_n) - f(t)| > \epsilon_0 \quad \textbf{for all } \quad t \in [s_n, s_n + \delta]. 
	\]
\end{definition}

The sequence \(s_n\) is referred to as a  \textit{divergence sequence}, while the intervals \([s_n - \delta, s_n + \delta]\) are known as \textit{divergence intervals}.

The existence of a sequence \(t_n\) is called \textit{Poisson stability}, and the existence of a sequence \(s_n\) and constants \(\epsilon_0, \delta\) is called the \textit{alpha unpredictability property}.

The main object to analyze in the paper is the integral $$F(t) = \int_{t_0}^{t} f(s)\,ds.$$ It is easily seen that the function is defined for all real numbers. The uniform continuity of the function follows from the boundedness of the integrand.


\begin{theorem} \label{theorem} Assume that a function $f(t)$ is alpha unpredictable. The integral $F(t)$ is alpha unpredictable if and only if it is bounded. 
\end{theorem}
\textit{Proof.} The necessity of the assertion is an immediate consequence of Definitions \ref{poisson}, \ref{alpha}.

The uniform continuity of the function follows from the boundedness of the integrand. In the following , we will apply the characteristics of the function \( f(t) \). Specifically, we will consider sequences \( t_n \) and \( s_n \), as well as the positive constants \( \varepsilon_0 \) and \( \delta \).

Firstly, we shall verify the sufficiency for existence of an convergence sequence, and for that we combine ideas of our \textit{method of intervals} \cite{Akhmet} with those of the \textit{sup} and \textit{inf} approach applied to a similar task in the book \cite{Bohr}. 

Let us  demonstrate that the convergence sequence $t_n$ is common for the both functions. 

Assume that \(F(t)\)  is bounded on the axis. That  is,  $ \inf_{-\infty<t<\infty} F(t)$ and $\sup_{-\infty<t<\infty} F(t)$
are real numbers. Denote them $m$ and $M$ respectively. Obviously, $m <M,$ since we plan to consider a non-trivial case. 

Let us fix a positive $\varepsilon$ and an interval $[a,b]$, as it is requested in Definition \ref{poisson}. According to the definitions of $\sup$ and $\inf$, there exist two points \(t_1\) and \(t_2\) such that
\begin{equation}
	F(t_1) = M - \frac{\varepsilon}{8},
\end{equation}
\begin{equation}
	F(t_2) = m + \frac{\varepsilon}{8}.
\end{equation}

To apply  the  method of included intervals,  let us fix the minimal-length interval $[c,d]$ that admits  the points $t_1,t_2, a,b$; consequently, we shall use it as the \textit{including interval}.    Next, consider  a large index $n$ such that 

\[|f(t+t_n)-f(t)| < \frac{5\epsilon}{8(d-c)},\]
if $t$ is between $c$ and $d$.
Then evaluate that 

\[  F(t_1+t_n)-F(t_2+t_n) = F(t_1) + \int_{t_1}^{t_1+t_n}f(s)ds -F(t_2)- \int_{t_2}^{t_2+t_n}f(s)ds \ge  \] 
\[(M - \frac{\varepsilon}{8}) -(m + \frac{\varepsilon}{8})  - \int_{t_1}^{t_2}(f(s+t_n)-f(s))ds > M-m -   \frac{\varepsilon}{4} - \frac{\varepsilon(d-c)}{4(d-c)}=\]
\[M-m - \frac{\varepsilon}{2}.\]
Rewrite the last inequality as 
\[(M -  F(t_1+t_n)) + (F(t_2+t_n)-m) < \frac{\varepsilon}{2}.\]
Now, we have that  for all arguments from the including interval, it is true that 
\[F(t+t_n)-F(t) = F(t_1+t_n) + \int_{t_1+t_n}^{t+t_n}f(s)ds -F(t_1)- \int_{t_1}^{t}f(s)ds =F(t_1+t_n) -F(t_1)+\] \[ \int_{t_1}^{t}(f(s+t_n)-f(s))ds > M - \frac{\varepsilon}{2} - M + \frac{\varepsilon}{8} - \frac{5\varepsilon(d-c)}{8(d-c)} = -\varepsilon. \]
Similarly, one can find that 
\[F(t+t_n)-F(t) = F(t_2+t_n) + \int_{t_2+t_n}^{t+t_n}f(s)ds -F(t_2)- \int_{t_2}^{t}f(s)ds =F(t_2+t_n) -F(t_2)+\] \[ \int_{t_2}^{t}(f(s+t_n)-f(s))ds < m + \frac{\varepsilon}{2} - m - \frac{\varepsilon}{8} + \frac{5\varepsilon(d-c)}{8(d-c)} = \varepsilon. \]
The last two inequalities imply that  $|[F(t+t_n)-F(t)| < \varepsilon$ for all arguments of the including interval, and consequently of the included interval $[a,b].$ It   has been approved that  $t_n$ is a convergence sequence for the integral-function.

Now, we shall  find a divergence sequence and other characteristics to finalize the proof. 
Fix an index $n,$  and consider two alternative  cases: $$(i) |F(s_n+t_n) - F(s_n)| \ge \frac{\varepsilon_0 \delta}{2},$$ and $$(ii) |F(s_n+t_n) - F(s_n)| <  \frac{\varepsilon_0 \delta}{2}.$$  In what follows, we shall show that for the integral-function the amplitude of divergence can be taken $\frac{\varepsilon_0 \delta}{4},$ the length of the interval divergence equals to $h= \frac{\varepsilon_0 \delta}{8M},$ and the element  $s'_n$ of the sequence for divergence moments can be chosen equal to $s_n$ or $s_n +\delta -h,$ depending on the two cases above. 

$(i).$ Assume that $|F(s_n+t_n) - F(s_n)| \ge \frac{\varepsilon_0 \delta}{2},$  and consider 
\[ F(t+t_n) - F(t) = F(s_n+t_n) - F(s_n) +  \int_{s_n}^{t}(f(s+t_n)-f(s))ds.\]
Denote $h= \frac{\varepsilon_0 \delta}{8M},$ and evaluate the last difference  on the interval 
$[s_n,s_n +h],$ 
\[ |F(t+t_n) - F(t)| >   \frac{\varepsilon_0 \delta}{2} -  2Mh = \frac{\varepsilon_0 \delta}{4}.\]
Thus, for the choice we have that the interval of divergence is $[s_n,s_n +h],$ and amplitude of the divergence is $ \frac{\varepsilon_0 \delta}{4}.$

$(ii).$  Let  the following inequality be correct:$|F(s_n+t_n) - F(s_n)| <  \frac{\varepsilon_0 \delta}{2}$  for the fixed index.  First, we consider the equality
\[ F(s_n+ \delta+ t_n) - F(s_n+\delta)  =  F(s_n+t_n) - F(s_n)+  \int_{s_n}^{s_n+\delta}(f(s+t_n)-f(s))ds.\]
Apply the equality to obtain that 
\[ |F(s_n+ \delta+ t_n) - F(s_n+\delta)| > \varepsilon_0 \delta -   \frac{\varepsilon_0 \delta}{2} = \frac{\varepsilon_0 \delta}{2}.\]
Now, use the equality,
\[F(t+t_n) - F(t) = F(s_n+ \delta+t_n) - F(s_n+ \delta) +  \int_{s_n+ \delta}^{t}(f(s+t_n)-f(s))ds, \]
to evaluate on the interval $[s_n+ \delta-h,s_n+\delta],$
\[  |F(t+t_n) - F(t)| > \frac{\varepsilon_0 \delta}{2} -   2Mh = \frac{\varepsilon_0 \delta}{4}. \]
It has been shown that the interval of divergence is $[s_n+ \delta-h,s_n+\delta],$ and amplitude of the divergence is $ \frac{\varepsilon_0 \delta}{4}.$

Thus, the characteristics for alpha unpredictability have been obtained for the fixed index. Since they do not depend on the index, we can confirm that the property is checked for the integral-function. 
The theorem is proved. 

\begin{corollary} Assume that $f(t)$ is an alpha unpredictable function, and 
	\[\int_{t_0}^{t}f(s)ds = ct + g(t),\]
	where $c$ is a real constant. If  $g(t)$  is a bounded function, then  it is alpha unpredictable one.		
\end{corollary}

\textit{Proof.}  One can write that 
\[\int_{t_0}^{t}(f(s)-c)ds = ct_0 + g(t).\]	
The difference  $f(t) -c$ is alpha unpredictable function \cite{Akhmet}, and  the sum $ct_0 + g(t)$ is bounded. So, by the theorem,  the last function is unpredictable. Consequently, the function $g(t)  = (ct_0 + g(t))  - ct_0$ is  alpha unpredictable \cite{Akhmet}. Here we have applied the algebra for alpha unpredictable functions developed in our former papers.   The corollary is proved.


\section{Discussion}  The study of recurrent functions is an excellent source for formulating problems in function theory and functional analysis, predominantly focusing on periodic and almost periodic functions. In our research, we introduce a new class of recurrency called \textit{alpha unpredictable functions}. Consequently, we recognize the need to develop a new theoretical framework and are actively pursuing this direction.

We have already established the algebraic properties of alpha predictable functions in our published works and created algorithms to generate new members of this class with hybrid properties. These hybrid functions combine features of both regular and chaotic behaviors, and we refer to them as \textit{compartmental alpha unpredictable functions}. The foundational basis for this class consists of periodic or almost periodic functions.

Additionally, we have explored extra conditions for periods and convergence sequences that ensure the alpha unpredictability of such composite functions. Although some elements of this emerging theory have been identified, numerous challenges still remain unresolved. In this article, we address one of them. We prove  the theorem on  integral of an unpredictable function, which is notable since a similar theorem already exists for periodic, quasi-almost-periodic, and Poisson stable functions. Thus, we continue to advance the development of the theory for recurrent functions. Moreover, we demonstrate a new technique of analysis for Poisson stability, more effective than those known in literature. 

The next most relevant assertion  in this series is  the alpha unpredictability of the derivative.  We also find the problem of Fourier series for the functions under examination to be particularly interesting, along with the study of the structure of convergence and divergence sequences, and the theory of functions on groups. Knowing that the alpha unpredictable function is a criterion for ultra-Poincare chaos, we are predicting amazing results to make fundamental contributions to the new development of the theory of chaos and how it turns out also for  the theory of functions and operators. Progress in this direction will open up new horizons in the development of applied mathematics. Given  complexity of the processes in our research, we understand that the proposed results can serve in the future for such applied research as quantum computers and also artificial intelligence.	


\begin{thebibliography}{99.}




\bibitem{PoinRec} Poincar\'{e}, H.;  Sur le probl\'{e}me destrois corpsetles \'{e}quations de la dynamique. {\em Acta Math.} {\bf 1890},  {\em 13}, 1--270.

\bibitem{Poin}  Poincar\'{e}, H.;  Les methodes nouvelles de la mecanique celeste. {\em Paris: Gauthier-Villars} {\bf 1892}, Vol. 1, 2. 


\bibitem{Birkhoff} Birkhoff, G. Dynamical Systems. {\em American Mathematical Society, Providence, RI, USA}, 1927.


\bibitem{Bohl}  Bohl, P.;  Über eine Differentialgleichung der Störungstheorie.{\em Grelles J.}, \textbf{ 1900}, {\em 131}, 268--321.

\bibitem{Esclangon}  Esclangon, E.;Les Fonctions Quasi-Periodiques. {\em Paris: Gauthier-Villars}, {\bf 1904}.  

\bibitem{Bohr} Bohr H.; Almost Periodic Functions. {\em Chelsea Publishing Company, New York}, \textbf{1947}.

\bibitem{LZ} Levitan, B.M.; Zhikov, V.V.;  Almost Periodic Functions and Differential Equations. {\em Cambridge: Cambridge University Press},\textbf{ 1982}, pp. 328.

\bibitem{Shch} Shcherbakov  B.A.: Poisson stable solutions of differential equations, and
topological dynamics. {\em Differential Equations}, \textbf{5}, \textbf{1969}, 2144–2155 .


\bibitem{Stepanov} Nemyckii, V.V and Stepanov, V.V.;  Qualitative Theory of Differential Equations (Russian), {\em Moscow, OGIZ}, \textbf{1947}.

\bibitem{sell} Sell, G. R.;  Topological Dynamics and Ordinary Differential Equations. {\em Van Nostrand Reinhold Company, London}, \textbf{1971}.


\bibitem{AkhmetFen} Akhmet M. and  Fen M.O. Unpredictable points and chaos. {\em Journal Communication in Nonlinear Science and Numerical Simulation}, {\bf 2016},   {\em 40}, 1--5.

\bibitem{Akhmet}  Akhmet M.;  Ultra Poincar\'e chaos and alpha labeling.{\em IOP, Birmingham}, \textbf{2025}.

\bibitem{kagan1} Akhmet M. U.,Baskan K. and Yesil C.:  Delta synchronization of {P}oincar\'e chaos in gas discharge-semiconductor systems. {\em Chaos}, \textbf{2022},  {\em 32(81)}, 083137.


\end{thebibliography}
\end{document}